\documentclass[twocolumn,10pt]{article}
\usepackage{graphicx} 
\usepackage{float}
\usepackage{amsthm}
\usepackage{amsmath}
\usepackage{amssymb}
\usepackage{mathtools}
\usepackage{algpseudocode}
\usepackage{courier}
\usepackage{geometry}
\usepackage{hyperref}

\newtheorem{lemma}{Lemma}[subsection]
\newtheorem{theorem}{Theorem}[subsection]
\newtheorem{definition}{Definition}[subsection]

\begin{document}

\title{Selection Improvements on the Parallel Iterative Improvement Algorithm for Stable Matching
\thanks{This work is proudly supported by NSF grant CNS-2149591. \\ GitHub \url{https://github.com/salberi1/Parallelizing_Stable_Matching}}}
\author{
  Scott Wynn$^{1}$, Alec Kyritsis$^{2}$, Stephora Alberi$^{3}$, Enyue Lu$^{3}$\\
  \\
  $^{1}$Department of Computer Science and Engineering, University of Washington, Seattle, WA 98195, USA\\
  $^{2}$Department of Computer Science, Middlebury College, Middlebury, VT 05753, USA\\
  $^{3}$Department of Computer Science, Salisbury University, Salisbury, MD 21801, USA
}

\date{}
\maketitle


\begin{abstract}
Sequential algorithms for the Stable Matching Problem are often too slow in the context of some large scale applications like switch scheduling. Parallel architectures can offer a notable decrease in runtime complexity. We propose a stable matching algorithm using $n^2$ processors that converges in $O(n log(n))$ average runtime. The algorithm is structurally based on the Parallel Iterative Improvement (PII) algorithm, where we improve the convergence rate from 90\% to 100\% over a large number of trials. 

We suggest alternative selection methods for pairs in the PII algorithm, called Right-Minimum and Dynamic Selection, as well as a faster preprocessing step, called Quick Initialization, resulting in full convergence over 3.6 million trials and significantly improved runtime.
\end{abstract}

\section{Introduction}

Originally developed to model the college admissions process, the Stable Matching Problem \cite{gale1962college} has spawned numerous applications in the social, physical and computational sciences \cite{fenoaltea2021stable, floreen2010almost, irving1987efficient, cao2024heuristic, roth1995nrmp}. One such application, switch scheduling in large data centers, is of particular interest due to the exponential increase in network throughput over the recent decade \cite{zhang2017stable}.
\\ \\
In the context of these large scale implementations, the classic Gale-Shapley algorithm from \cite{gale1962college} can be too slow, requiring $O(n^2)$ runtime to complete. Parallel and Distributed architectures can offer a notable decrease in the computational complexity of the Stable Matching Problem \cite{feder2000sublinear, huang2007on, subramanian1992parallel}. Without considering some of the practical constraints on parallel architecture, \cite{feder2000sublinear} showed that a sublinear algorithm for stable matching is possible. In this paper, we study an augmented version of the Parallel Iterative Improvement (PII) algorithm proposed in \cite{lu2003parallel}, which improves on previous ideas in \cite{white2013parallel} and preliminary work in \cite{kyritsis2023dynamic} and offers significant runtime improvement while requiring only $O(n^2)$ processors, making the parallel architecture required for the algorithm much more feasible for large scale applications.
\\ \\
In its initial implementation, the PII algorithm finds a stable matching within $5n$ iterations and a total of $O(n \: log(n))$ runtime in approximately $90\%$ of cases when tested on input size of $n = \{10, 20, ..., 100\}$. In the remaining $10\%$ of cases, the algorithm cycles indefinitely, repeatedly returning to a previous iteration. We show that augmenting the iteration methods to force convergence can significantly improve runtime and convergence rate of the PII algorithm. 

\subsection{Preliminaries}
Denote the sets of $n$ men and $n$ women by $M = \{m_1, m_2, ..., m_n\}$ and $W = \{w_1, w_2, ..., w_n\}$ respectively. Let $P$ be the $n \times n$ matrix, or \textit{preference matrix}, with each row i corresponding to the ith element of $M$ and each column j corresponding to the jth element of $W$, as used in \cite{lu2003parallel}. Hence, $p_{i, j} \in P$ represents the respective rankings of man $i$ with woman $j$. For a given pairing $p_{i, j}$, denote the preference ranking of man $i$ for woman $j$ and the preference ranking of woman $j$ for man $i$ by $L(p_{i, j})$ and $R(p_{i, j})$ respectively. We call these quantities the \textit{left value} and \textit{right value} of $p_{i, j}$ respectively. Then, man $i$ prefers woman $m$ to woman $j$ if $L(p_{i, m}) < L(p_{i, j})$, and similarly woman $j$ prefers man $l$ to man $i$ if $R(p_{l, j}) < R(p_{i, j})$.
\\ \\
Let $\mu$ be a one-to-one mapping between $M$ and $W$. We call $\mu$ a \textit{matching} and say that $p_{i, j}$ is \textit{matched} under $\mu$ if and only if $\mu(i) = j$ and $\mu(j) = i$. For convenience, we often say $p_{i, j} \in \mu$. For a given matching $\mu$, if man $i$ and women $j$ prefer each other to their current partners, then $p_{i, j}$ is a \textit{blocking pair} and the current pairs are called \textit{unstable pairs}. Formally, in a preference matrix $P$, a pair $p_{i, j}$ is a blocking pair with respect to a matching $\mu$ and unstable pairs $p_{i, l}, p_{m, j} \in \mu$ if and only if $L(p_{i, j}) < L(p_{i, l})$ and $R(p_{i, j}) < R(p_{m, j})$. If a matching $\mu$ on preference matrix $P$ contains any unstable pairs, we classify $\mu$ as an unstable matching.
\\ \\
Finally, let $\mu_1, \mu_2, ..., \mu_n$ denote a set of matchings on preference matrix $P$. Suppose $\mu_k$ is an unstable matching with blocking pair $p_{i, j}$. Then $\mu_{k + 1}$ \textit{satisfies} $p_{i, j}$ if and only if $p_{i, j} \in \mu_{k + 1}$. In the following sections it will be useful to track general sequences of blocking pairs.

\subsection{Parallel Architecture}
We assume $n^2$ processors (PE's) with hypercube or multiple broadcast bus architecture. This enables row and column broadcasts and find minimum operations to be completed with at most $O(log(n))$ complexity, as outlined in \cite{lu2003parallel}. It may be useful to consider the PE's arranged on an $n \times n$ mesh $A$ such that $PE_{i, j}$ corresponds to pair $p_{i, j}$ in preference matrix $P$.

\subsection{Statement of Results}
Our contribution to the problem lies in the presentation and analysis of two new pair selection methods to augment the PII algorithm.
\\ \\
Method 1: We propose Right-Minimum Selection, a method to force termination of the algorithm by only choosing blocking pairs in which one side's preferences improve.
\\ \\
Method 2: We propose Dynamic Selection, a method that optimally chooses the subsequent matching by considering all previously selected blocking pairs as a strong combination.
\\ \\
The new augmented algorithm converges faster and more often than previous iterations of the PII algorithm (showing $100\%$ convergence over $3.6$ million trials across different n values), and the convergence rate scales significantly better with higher values of n. 

\subsection{Overview}
The rest of the paper will proceed as follows. In section 2 we will discuss related work, including a brief overview of the PII Algorithm and previous iterations on the PII Algorithm. In Section 3 we present our selection methods, Right-Minimum and Dynamic Selection. Section 4 presents the proposed augmented PII algorithm, and Sections 5 and 6 analyze the results of our proposed methods and algorithm and outline future work.

\section{Previous Work}

\subsection{Iterative Stable Matching Algorithms}
When observed in real life through social networks, matchings attempt to converge toward stability through pairing together blocking pairs. Roth and Vande Vate have shown that swapping along blocking pairs at random from any initial matching will almost surely result in a stable matching eventually \cite{roth1990random}.
However, it has also been shown that iterative methods for swapping along specific blocking pairs each iteration may result in indefinite cycling, resulting in a failure to ever converge to stability \cite{knuth1976marriages}.

\subsection{The PII Algorithm}
The PII Algorithm presented in \cite{lu2003parallel} is one such replacement algorithm.
\\ \\
The PII algorithm begins with a randomly generated matching $\mu_0$. The algorithm then executes the following steps each iteration, terminating when a stable matching is achieved:
\\ \\
1) Find all blocking pairs
\\ \\
2) In each row of M with a blocking pair, select the one with the lowest left value as an NM1-generating pair
\\ \\
3) In each column of M with an NM1-generating pair, select the one with the lowest right value as an NM1 pair
\\ \\
4) Remove all pairs in the current matching in the same row or column as an NM1 pair, and add all NM1 pairs to the matching
\\ \\
5) Fill in any open columns and rows with pairs as described in \cite{lu2003parallel}.
\\ \\
 The PII algorithm requires $n^2$ parallel processors, and each iteration will complete in $O(log(n))$ complexity. Lu empirically observed that the convergence rate of the PII algorithm on randomly generated preference matrices steadily decreases from $99\%$ at $n=10$ to $86\%$ at $n=100$. In cases where the PII algorithm did not converge, it continued cycling indefinitely.

\subsection{The PII-SC Algorithm}
White proposed the Smart Initialization and Cycle Detection methods to improve the convergence of the PII Algorithm in \cite{white2013parallel}.
\\ \\
In Smart Initialization, each man $m_i$ in $M$ proposes to his top choice woman. If multiple men propose to the same women, she is matched with the man she most prefers. Those men left unmatched then propose to their next choice woman, excluding any women matched in a previous round, and a similar process ensues. This occurs until all men have been matched with a woman, and runs with $O(n \: log(n))$ complexity in parallel as described in \cite{white2013parallel}.
\\ \\
White observed that the majority of cycles were formed by a single $NM1$ pair was alternating between pairs in two distinct sub stable matchings each iteration. A stable matching could then be constructed by including every other $NM1$ pair in a cycle and the set of pairs that remained in the matching throughout the cycle. Some more prominent edge cases for detecting cycles were also added to the cycle detection method proposed in \cite{white2013parallel}.
\\ \\
The PII-SC Algorithm combines the PII Algorithm with Smart Initialization and Cycle Detection, and fails to converge in approximately $1$ in $1$ million randomly generated preference matrices from $n=10$ to $n=100$, with failure to converge significantly higher than $1$ in $1$ million for $n=100$, suggesting some scalability issues. Due to the existence of numerous cases where cycles were not formed as described above, the PII-SC algorithm and subsequent improvements have continued improving the observed convergence by patching edge cases. As a result, the convergence rate of the PII-SC algorithm decreases significantly on larger $n$, as more frequent edge cases occur, illustrated by the results in \cite{white2013parallel}.

\begin{figure*}[t]
    \centering
    \includegraphics[scale = 0.65]{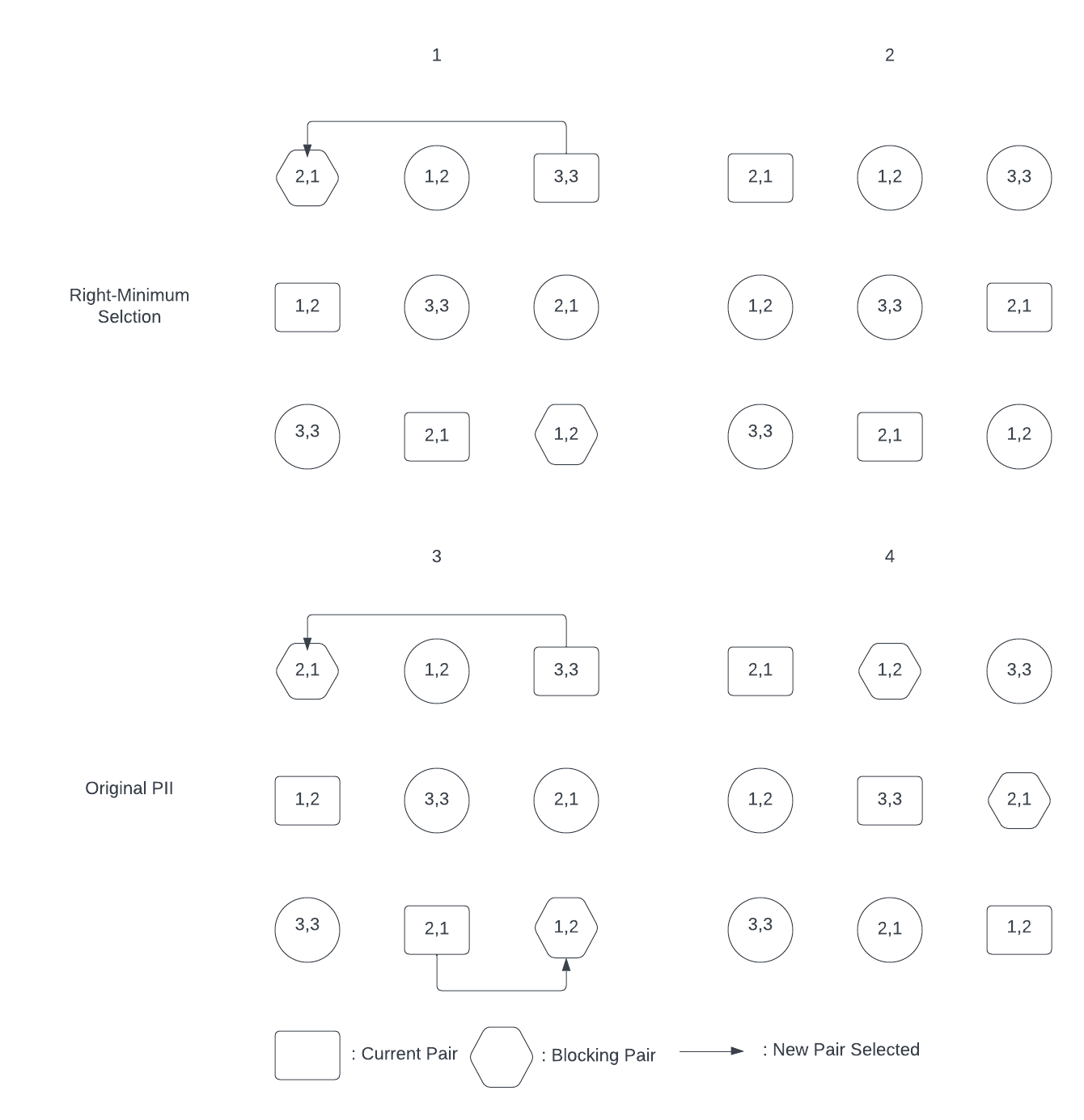}
    \caption{Comparison of a single iteration of Right-Minimum Selection (from 1 to 2) and a single iteration of the original PII algorithm (from 3 to 4). Each sub-figure (1,2,3,4) contains the entries of the preference matrix, with shapes indicating the matching at a given iteration. The restriction imposed by Right-Minimum selection allows immediate convergence, while PII continues to allow multiple blocking pairs.}
    \label{fig:nm1Right-Selection}
\end{figure*}

\section{Selection Methods}

In this section, we present our main contribution to the line of PII algorithms: Updated selection methods.

\subsection{Right-Minimum-Selection}

\noindent To enhance the algorithm's efficiency, we introduce a novel selection method termed Right-Minimum-Selection. We define \textbf{Right-Minimum-Selection} as follows:

\begin{definition}
    [Right-Minimum-Selection] Let $S$ be a preference matrix with $n$ men and $n$ women, and let $\mu$ be an unstable matching on $S$. Suppose $p_{i, j}$ is an unstable pair on $\mu$ with corresponding blocking pair $p_{i, l}$. Then the algorithm will consider $p_{i, l}$ as a potential NM1-generating pair if and only if $R(p_{i, j}) > R(p_{i, l})$.
\end{definition}

\noindent Informally, we only select NM1-generating pairs where for a given row the woman prefers her new matching, demonstrated in \ref{fig:nm1Right-Selection}.
\\ \\

\noindent Often, we observed that the PII algorithm keeps unstable pairs in the matching by satisfying a blocking pair that does not improve the matching overall, leading to cycles as described in \cite{lu2003parallel}. Intuitively, Right-Minimum selection forces the algorithm to terminate by continually iterating on the matching in favor of women similar to the process in a similar manner to the original Gale-Shapley algorithm. It should be noted that a Left-Minimum selection defined in a similar manner would yield similar results (with random initialization).

\hfill \break
\noindent Right-Minimum Selection leads to the following results justifying the algorithm will terminate:

\begin{lemma}
    [$NM1$-Cycle Freeness] Let $A$ be an instance of the $PII$ algorithm. If $A$ uses Right-Minimum-Selection over its entire duration, then $A$ is $NM1$-cycle free. 
\end{lemma}

\begin{proof}
  Let $S = \{p_1, p_2, ..., p_k\}$ be a sequence of NM1-pairs such that $p_{i+1}$ is a blocking pair of $p_i$. Observe that $p_{i + 1}$ may only replace $p_i$ if it appears in an identical row or column. If replacement occurs in a column, then $R(p_{i + 1}) < R(p_i)$ by the definition of a blocking pair. If replacement occurs in a row, $R(p_{i + 1}) < R(p_i)$ by the definition of Right-Minimum Selection. Thus, we produce a chain of NM1 pairs $R$ satisfying $R(p_1) > R(p_2) > ... > R(p_k)$. For a cycle to occur, $p_k$ must return to a row or column of $p_1$, violating the monotonicity of $R$ and the results follow. 
\end{proof}

\noindent We also have

\begin{lemma}
    Let $S = \{p_1, p_2, ..., p_k\}$ be an $NM1$-path taken by the Algorithm. Then each time $S$ is traversed, it's length decreases by at least two.
\end{lemma}

\begin{proof}
    Let $P$ be an $n \times n$ preference matrix and $S$ be an $NM1$ path such that $S = \{p_1, p_2, ..., p_k\, \; p_i \in S\}$. Assume we traverse $P$ once to arrive at $v_k$. Observe that by definition $v_k$ may not be replaced by an NM1 or NM2 pair. Since $v_{k - 1}$ lives in a row or column of $v_k$, upon a second traversal of $S$ we may not arrive at $v_{k -1}$ since it would have to replace $v_k$. Now define $S' = S \setminus \{v_k, v_{k -1}\}$ and observe that $S'$ is also an $NM1$ path. We then apply similar logic from above to show that $P'$ must also satisfy this property. Hence, we may construct a sequence of paths $P \supset P' \supset P'' \supset ...$ such that the cardinality of each path decreases by at least two upon each traversal.
    
\end{proof}

\noindent We combine Lemmas 3.1.1 and 3.1.2 to produce the following result:

\begin{theorem}
    [Cycle Freeness]
    Let $A$ be an instance of the $PII$ algorithm. If $A$ uses Right-Minimum Selection over its entire duration, then $A$ is cycle free.
\end{theorem}

\begin{proof}
    Assume the contrary. Lemma 1 implies that such a cycle must be composed of both $NM1$ and $NM2$ pairs. Consider the $NM1$ path $S$ that takes the shortest iterations to traverse in the cycle with length $k$. Suppose $S$ spawns $NM2$ pair $u_i$ at $NM1$ pair $p_i \in S$. By lemma 2 and the assumption of cycling, we may traverse $S$ till $v_i$ is deleted. Then $p_i$ may not spawn $u_i$, contradiction.
    
\end{proof}

\noindent \textbf{Complexity}

\noindent It is straightforward to see that the additional check in each step when choosing NM1-generating pairs can be achieved in $O(1)$ time, and thus the per iteration time complexity remains $O(log(n))$.
\\ \\
We note that while Right-Minimum Selection has the strong property that it may not cycle, in practice it does not necessarily yield a stable matching and often is supplemented with a standard $NM1$-selection method. This occurs when all blocking pairs remaining in a matching do not satisfy the condition to be considered under Right-Minimum Selection.
\\ \\
We have empirically determined that, if Right-Minimum Selection terminates without successfully finding a stable matching, there are very few blocking pairs remaining and continuing with the original PII selection method often finds a stable matching.

\subsection{Dynamic Selection}

\begin{figure*}[t]
    \centering
    \includegraphics[scale = 0.65]{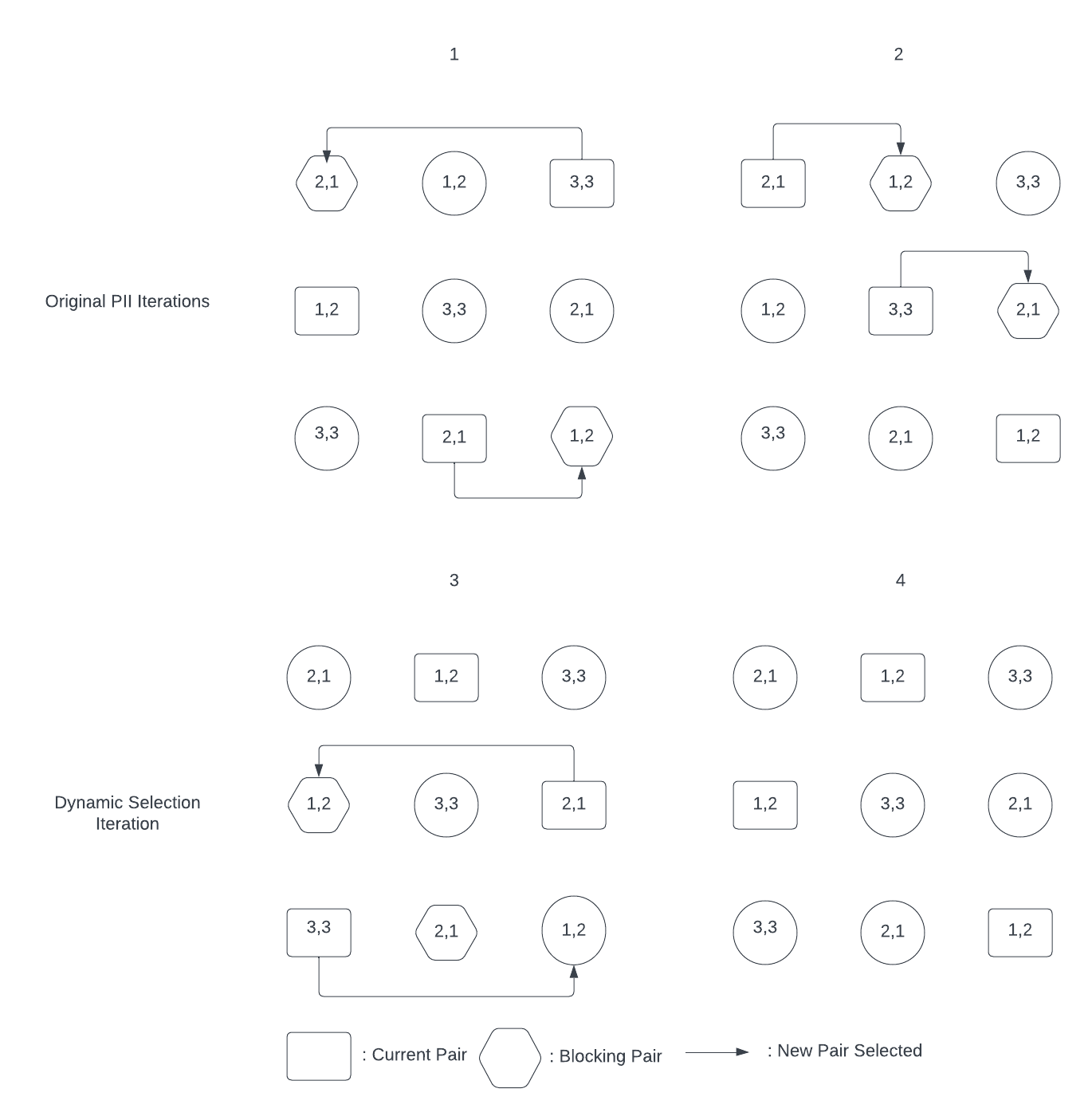}
    \caption{Example of Dynamic Selection. Each sub-figure (1,2,3,4) contains the entries of the preference matrix, with shapes indicating the matching at a given iteration. Iterations of the original PII algorithm are done for the first 2 iterations (from 1-2 and 2-3). Dynamic Selection is done for the 3rd iteration (from 3-4) leading to convergence. Note that another iteration of the original PII algorithm after the matching at 3 would return to the initial matching (1) causing a cycle.}
    \label{fig:Dynamic-Selection}
\end{figure*}

\noindent We define \textbf{Dynamic Selection} as follows:

\begin{definition}
    [Dynamic Selection] Let $S$ be a preference matrix with unstable matching $\mu_k$. Assume NM1 selection as proposed in the original PII algorithm in \cite{lu2003parallel}, and let $N_i$ be the set of NM1 pairs in row $i$ the algorithm chooses over matchings $\mu_1, \mu_2, ..., \mu_{k - 1}$. Suppose $p_{i, j}$ is a blocking pair of $\mu_k$. Then $p_{i, j}$ is an NM1-generating pair if and only if $L(p_{i,j}) < \min\{L(p), p \in N_i\}$.

\end{definition}

    \noindent Intuitively, in an iteration where Dynamic Selection is used, we only select a blocking pair as a new  NM1-generating pair if its left value is less than all previously selected NM1-generating pairs, as shown in  \ref{fig:Dynamic-Selection}.
    \\ \\
\textbf{Complexity}

\noindent We now justify that the per iteration time complexity remains $O(log(n))$ with Dynamic Selection:

\hfill \break
\noindent Each $PE$ in row $i$ maintains a pointer to $PE_{i, l}$ corresponding to $L(p_{i, l}) = \min\{L(p), p \in N_i\}$. For convenience we term this the \textit{minimum pointer} and $PE_{i, l}$ the \textit{left-minimum processor}. $PE_{i, l}$ points to itself. To avoid instances in which the algorithm consecutively selects that pair with the minimum left value, and to facilitate discovery of a greater number of NM1 pairs, we impose a wait time $\mathcal{W} \in \mathbb{Z}^+$.  Hence, $PE_{i, l}$ also logs the last iteration at which it was compared $c$, and may enter the comparison process again when $k - c > \mathcal{W}$ where $k$ is the current iteration. If this is the case, we say the \textit{wait condition} is satisfied.
\\ \\
Intuitively, the wait condition ensures that we only use dynamic selection when we have enough previously selected pairs to make the step effective, and when the standard iterations are taking too long to converge, suggesting the algorithm is likely in a cycle.
\\ \\
It is straightforward to see that the additional variables may be stored with $O(1)$ space complexity. We show that per iteration time complexity remains $O(log(n))$:

\begin{proof}

    Assume at iteration $k$ we have unstable pair $p_{i, j}$ in row $i$. Dynamic Selection may be achieved with the following steps:
    
    \begin{enumerate}
        \item If $N_i = \emptyset$ and $p_{i, j}$ is selected as an NM1 pair, proceed to 4. Otherwise, continue.
        \item If the wait condition is satisfied, left-minimum processor $PE_{i, l}$ enters the comparison procedure outlined in section 5.2. Otherwise, NM1 selection proceeds as normal. Row-wise find-minimum operations may still be achieved in $O(log(n))$ time.
        \item If $p_{i, j}$ is selected as an NM1 pair, $PE_{i, j}$ compares its left value to $PE_{i, l}$. This is achieved in $O(1)$ time. If $L(p_{i, j}) \ge L(p_{i, l})$, no change occurs. Otherwise, $L(p_{i, j}) < L(p_{i, l})$ and proceed to 4.
        \item  $PE_{i, j}$ sets its minimum pointer to itself, and broadcasts its position along row $i$. Each processor then proceeds to set its minimum pointer to $PE_{i, j}$. This is achieved in $O(log(n))$ time.
        
    \end{enumerate}

    \noindent Each step occurs in at most $O(log(n))$ time. Hence, per iteration complexity remains $O(log(n))$. 

\end{proof}

\noindent We have empirically determined that incrementing $\mathcal{W}$ each time the left-minimum processor is selected as an NM1 pair yields the best results. We then set $\mathcal{W}$ to 2 when a new a left-minimum processor is chosen.

\subsubsection{Preprocessing}
We also propose a new quick initialization method. Each man $m_i$ in $M$ sequentially proposes to his top choice woman remaining. When a woman receives a proposal, they are matched and the woman is removed from the preference lists of all other men. After all men have proposed once, all men will have been matched.
Each woman receives exactly one proposal, so proposals can be done in $O(1)$ runtime. In parallel with $n^2$ processors, removing a woman from the preference lists of other men can also be done in $O(1)$ runtime, allowing quick initialization to run with $O(n)$ complexity in parallel.

\section{The PII-RMD Algorithm}

We define the PII-RMD algorithm as the PII algorithm augmented with Right-Minimum Selection and Dynamic Selection. We have empirically determined that beginning Right-Minimum Selection after $n$ iterations and initially setting the wait time $\mathcal{W}$ of Dynamic Selection to $n$ produced the fastest convergence.

\section{Results}

\begin{figure*}[t]
    \centering
    \includegraphics[scale = 0.65]{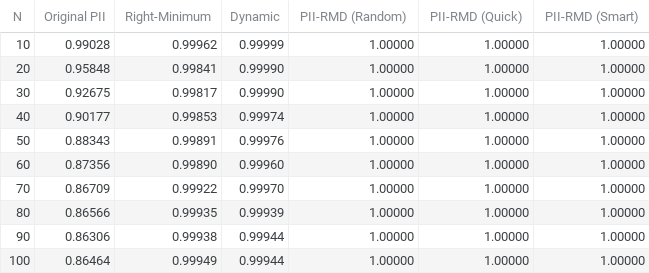}
    \caption{Probability of successfully finding a stable matching within $5n$ iterations with $100,000$ trials for various $n$ ranging from $10-100$. All methods were tested using random initialization, except PII-RMD (Quick) and PII-RMD (Smart) which used Quick and Smart Initialization respectively.}
    \label{fig:converge}
    \vspace{0.5cm}
    \centering
    \includegraphics[scale = 0.65]{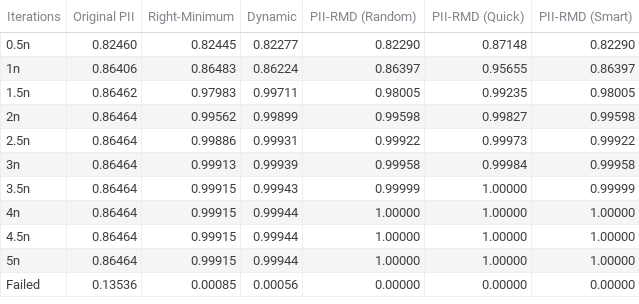}
    \caption{Probability of successfully finding a stable matching with $n=100$ with $100,000$ trials within various numbers of iterations ranging from $0.5n-5n$. All methods were tested using random initialization, except PII-RMD (Quick) and PII-RMD (Smart) which used Quick and Smart Initialization respectively.}
    \label{fig:speed}
\end{figure*}

\noindent We have implemented the complete PII-RMD algorithm, along with Right-Minimum Selection, Dynamic Selection, and the original PII algorithm for comparison. For each algorithm, we tested for convergence and speed using both randomized initialization and Smart Initialization.
\\ \\
In each trial, we create randomized preference lists, done by generating a random permutation of n preference lists of length n, and run each of the different algorithms with each initialization method on the created preference matrix. In \ref{fig:converge}, we computed the success rate of each algorithm in finding a stable matching within 5n iterations where n ranges from 10 to 100, running 1 million trials for each algorithm variant. In \ref{fig:speed}, we computed the number of iterations each algorithm required to find a stable matching at n = 100, varying the number of iterations from 0.5n to 5n and running 100,000 total trials for each algorithm variant.
\\ \\
From \ref{fig:converge}, using Right-Minimum Selection instead of the original PII selection improves convergence at larger n from $86\%$ to over $99.9\%$. Furthermore, Right-Minimum Selection does not see a performance decrease as n grows in \ref{fig:converge}. When tested with limited trials ($1000$ trials for each value of n) for n ranging from 100 to 1000, Right-Minimum selection also saw no performance decrease at higher values of n. Cases where Right-Minimum Selection alone fails to find a stable matching still terminate as shown in section 3.1, and the final matching has always been observed to contain very few unstable pairs.
\\ \\
Dynamic Selection and Cycle Detection both force pairs from a single stable matching to be chosen when the algorithm is cycling between pairs from multiple separate matchings. Comparing Dynamic Selection with the Cycle Detection data in \cite{white2013parallel}, both methods perform similarly in improving algorithm convergence rate. However, the cycle detection method does not begin until iteration 3n, causing a significant convergence runtime improvement when using Dynamic Selection over Cycle Detection.
\\ \\
For the same reason, the PII-RMD algorithm showed an improvement in convergence speed compared to the PII-SC algorithm (for example, the PII-RMD algorithm reaches $99.9\%$ convergence within 2.5n iterations with any initialization method, whereas the PII-SC algorithm with Smart Initialization requires 3.5n iterations to do the same). Overall, the PII-RMD algorithm converged in all cases for random, quick, and Smart initialization, with fastest convergence when using quick initialization \ref{fig:speed} removing the need for an expensive $O(n \: log(n))$ runtime for preprocessing.
\\ \\
To assess the scalability of our PII-RMD algorithm with higher values of n, we have also tested the PII-RMD algorithm for n=100, 110, ..., 200. For each initialization method we ran 10,000 tests for each value of n, resulting in 330,000 total trials which fully converged, suggesting the PII-RMD algorithm does not see the same steep decline with increasing $n$ observed in the PII-SC algorithm even for $n \leq 100$ in \cite{white2013parallel}. These results support the scalability of the PII-RMD algorithm, achieving the original PII algorithm's purpose of an efficient in-practice stable matching algorithm on a reasonable number of processors.
\\ \\
Overall, the PII-RMD algorithm is empirically fully convergent within 5n iterations in all 3,630,000 trials across all tested values of n, including testing at higher values of n compared to previous PII algorithm iterations. The runtime required for preprocessing and the number of iterations required for convergence were also reduced compared to previous PII algorithm iterations.

\section{Concluding Remarks and Future Work}

The results suggest the PII-RMD algorithm is a significant improvement over previous PII algorithm variations in both runtime and convergence. Empirically, the PII-RMD algorithm exhibits complete convergence within 5n iterations, and thus has always been observed to converge in $O(n \: log(n))$ runtime.
\\ \\
The ultimate goal of our research in the PII algorithm is to determine a fully convergent stable matching algorithm on $n^2$ processors with a theoretical worst-case runtime of $O(n \: log(n))$, to provide an efficient algorithm with a low requirement on the number of processes available for large scale applications. It remains unclear whether PII-RMD algorithm fully converges theoretically.
\\ \\
We hope to prove the convergence of the PII-RMD algorithm or another algorithm satisfying the same runtime and parallel architecture constraints. In doing so, we hope to also better understand the properties of instability for the stable matching problem, as very limited research has been done in this area \cite{eriksson2008instability}.
\\ \\

\end{document}